\documentclass[a4paper,11pt]{article}
\pdfoutput=1 

\usepackage{jheppub} 
\usepackage{amssymb,amstext,amsthm,amsfonts,mathrsfs,amsbsy,amsmath,array,multirow}
\usepackage[T1]{fontenc} 
\newtheorem{thm}{Theorem}[section]

\newtheorem{prop}[thm]{Proposition}
\newtheorem{lemma}[thm]{Lemma}
\newtheorem{cor}[thm]{Corollary}

\newtheorem{definition}[thm]{Definition}

\newtheorem{remark}[thm]{Remark}

\rightline{\small IPhT-t14/180}

\title{\boldmath M-theory moduli spaces and torsion-free structures}

\author{Mariana Gra\~na}
\author{and C. S. Shahbazi}

\affiliation{Institut de Physique Th\'eorique, CEA Saclay.}

\emailAdd{mariana.grana@cea.fr}
\emailAdd{carlos.shabazi-alonso@cea.fr}

\abstract{\vspace{0.1cm}\\ Motivated by the description of $\mathcal{N}=1$ M-theory compactifications to four-dimensions given by Exceptional Generalized Geometry, we propose a way to \emph{geometrize} the M-theory fluxes by appropriately relating the compactification space to a higher-dimensional manifold equipped with a torsion-free structure. As a non-trivial example of this proposal, we construct a bijection from the set of $Spin(7)$-structures on an eight-dimensional $S^{1}$-bundle to the set of $G_{2}$-structures on the base space, fully characterizing the $G_{2}$-torsion clases when the total space is equipped with a torsion-free $Spin(7)$-structure. Finally, we elaborate on how the higher-dimensional manifold and its moduli space of torsion-free structures can be used to obtain information about the moduli space of M-theory compactifications.}

\arxivnumber{1410.8617}

\begin{document}
\maketitle
\flushbottom


\section{Introduction}
\label{sec:introduction}


In the context of String/M-theory compactifications, supersymmetry implies the existence of a topological $G$-structure satisfying a particular set of differential conditions on the appropriate principal bundle over the internal manifold. In particular, this set of differential conditions can be written as first-order partial differential equations on the tensors that define the $G$-structure. The moduli space of a given supersymmetric compactification is an important space from the physics point of view; it is expected to be a finite-dimensional manifold equipped with a Riemannian metric closely related to the non-linear sigma-model appearing in the effective supergravity action of the compactification. In the simplest cases, namely in the absence of fluxes, the topological reduction is on the frame bundle of the internal space and the differential conditions implied by supersymmetry are equivalent to imposing the $G$-structure to be \emph{torsion-free} \cite{Candelas:1985en,Candelas:1984yd}, meaning that the reduction is not only topological but also geometric. Usually the moduli space of this torsion-free structures is well under control when the internal space is an oriented, compact manifold \cite{Candelas:1990pi,Joyce2007}. In particular, the Riemannian metric is known and thus it can be used to explicitly write the effective action of the compactification \cite{Ferrara:1988ff,Cecotti:1988qn}. 

From the phenomenological point of view, realistic scenarios require the moduli spaces to be just isolated points, and for that  
the presence of non-vanishing fluxes is crucial\footnote{It is also possible to obtain phenomenologically viable models through flux-less compactifications on manifolds with the appropriate kind of singularities, see references \cite{Acharya:2002kv,Acharya:2004qe}.}. This implies that the topological $G$-structure induced by supersymmetry will have non-vanishing torsion \cite{Strominger:1986uh,Gauntlett:2003cy}, always subject to the corresponding differential equations. It is thus desirable to have a better understanding of the moduli space in the presence of fluxes, in order to find as explicitly as possible the four-dimensionl models coming from String/M-theory compactifications, that can ultimately be used to make contact with particle physics or cosmology.  

A huge effort has been devoted so far in understanding the supersymmetric compactifications, see \cite{Grana:2005jc} and references therein. One approach, namely generalized geometry either in its complex \cite{2002math......9099H,2004math......1221G} or exceptional flavour \cite{Hull:2007zu,Pacheco:2008ps}, has proven to be particularly interesting and fruitful, giving a unified geometrical description of the different flux backgrounds as well as their moduli spaces \cite{Grana:2005ny,Grimm:2005fa,Grana:2006hr,Pacheco:2008ps,Grana:2009im,Grana:2012zn}.  

In this note we are going to argue that through the exceptional generalized geometric description of the moduli space of $\mathcal{N}=1$ supersymmetric M-theory compactifications to four-dimensions, we can associate to every internal seven-dimensional manifold $\mathcal{M}_{7}$ an eight-dimensional manifold $\mathcal{M}_{8}$ equipped with a particular set of tensors $\mathfrak{S}$ (given in Definition \ref{def:intermediatemanifolds}) that contain all the information about the topological reduction implied by supersymmetry. This eight-dimensional manifold $\left(\mathcal{M}_{8},\mathfrak{S}\right)$ was dubbed \emph{intermediate manifold} in \cite{Grana:2014rpa}, where the study of its topological properties was initiated. The geometry of an intermediate manifold $\left(\mathcal{M}_{8},\mathfrak{S}\right)$ is completely specified once the tensors $\mathfrak{S}$  are required to satisfy the differential conditions obtained from supersymmetry. In order to translate the supersymmetry conditions to the intermediate structure $\mathfrak{S}$, it is necessary first to write them in the exceptional generalized geometry language, which has not been done yet. Once they are known we will be able to fully uncover the role of intermediate manifolds in relation to M-theory compactifications and their moduli spaces. However, it is clear that the moduli space of intermediate structures on $\mathcal{M}_{8}$ is closely related to the moduli space of the corresponding compactification, and therefore we can resort to this guiding principle in order to construct particular examples of $\mathcal{M}_{8}$.

More precisely, here we will propose that it is natural to construct $\mathcal{M}_{8}$ from $\mathcal{M}_{7}$ in such a way that $\mathcal{M}_{8}$ is equipped with a torsion-free $Spin(7)$-structure induced by a one-to-one map from the $G_{2}$-structure in seven dimensions, obtaining therefore a candidate for a map between $G_{2}$-structures with non-zero torsion in seven dimensions and torsion-free structures in eight-dimensions. This is interesting because the moduli space of the latter is much more under  control than the moduli space of the former, and thus it is reasonable to expect that using the map some new information can be obtained about the moduli space of $G_{2}$-structures with non-vanishing torsion.

Once we realize that one can embed structures with torsion in a given dimension into torsion-free structures in a higher dimension, there is no need to stop at eight-dimensions. In fact, for a general enough $G_{2}$-structure in seven-dimensions we should not expect to obtain a torsion-free $Spin(7)$-structure in eight dimensions, but that could be possible if we consider instead a manifold of dimension large enough. This construction can be as well motivated from the Exceptional Generalized Geometry formulation of the compactification.

The outline of this note is as follows. In section \ref{sec:intermediateEGG} we review $\mathcal{N}=1$ M-theory compactifications to four dimensions,  motivate the definition of intermediate manifolds and show  how  they can be constructed from the physical manifold $\mathcal{M}_{7}$. 
In section \ref{sec:differentialconditions} we present  a few examples of this construction. In particular, and also as a non-trivial example of intermediate manifold, we obtain a bijection from the set of $Spin(7)$-structures on an eight-dimensional $S^{1}$-bundle and the set of $G_{2}$-structures on the base space, fully characterizing the $G_{2}$-torsion clases when the total space is equipped with an invariant torsion-free $Spin(7)$-structure.
In section \ref{sec:modulispaces} we define the moduli space of supersymmetric M-theory compactifications and explain how the map between $G_{2}$-structures in seven-dimensions and torsion-free structures in eight-dimensions can be used to obtain information about the former. 


\section{Intermediate manifolds from Exceptional Generalized Geometry}
\label{sec:intermediateEGG}


Let us consider the bosonic sector $\left(\mathsf{g},\mathsf{G}\right)$ of eleven-dimensional Supergravity on a space-time manifold $\mathcal{M}$ that can be written as a topologically trivial direct product

\begin{equation}
\label{eq:productmanifold}
\mathcal{M} = \mathcal{M}_{1,3}\times\mathcal{M}_{7}\, ,
\end{equation}

\noindent
where $\mathcal{M}_{1,3}$ is a four-dimensional Lorentzian oriented, spin manifold and $\mathcal{M}_{7}$ is a seven-dimensional, Riemannian, oriented, spin manifold. Although we will speak about \emph{compactifications}, we are not going to assume that $\mathcal{M}_{7}$ is compact since it might be consistent to compactify in non-compact space-times with finite volume and appropriate behavior of the laplacian operator \cite{Nicolai:1984jga}. 
According to the product structure (\ref{eq:productmanifold}) of the space-time manifold $\mathcal{M}$, the 
%
%
we choose the Lorentzian metric $\mathsf{g}$ on $\mathcal{M}$ to be given by

\begin{equation}
\label{eq:11dmetricproduct}
\mathsf{g} = e^{2A} \mathsf{g}_{1,3}\times \mathsf{g}_{7} \, , 
\end{equation}

\noindent
where $\mathsf{g}_{7}$ is a Riemannian metric on $\mathcal{M}_{7}$, $e^{2A} \in C^{\infty}\left(\mathcal{M}_{7}\right)$ is the warp factor and $\mathsf{g}_{1,3}$ is a Lorentzian metric on $\mathcal{M}_{1,3}$. In addition, and in accordance again with the product structure of $\mathcal{M}$ and the choice of metric \eqref{eq:11dmetricproduct}, we will assume that the four-form $\mathsf{G}$ is given by

\begin{equation}
\label{eq:Gproduct}
\mathsf{G} = \mu\, \mathrm{pr}^{\ast}_{1}\mathrm{Vol_{4}} + \mathrm{pr}^{\ast}_{2} G\, , \qquad \mu\in \mathbb{R}\, , \qquad G\in\Omega^{4}_{cl}\left(\mathcal{M}_{7}\right)\, ,
\end{equation}

\noindent
where $\mathrm{pr}_{1}\colon \mathcal{M}\to\mathcal{M}_{1,3}$ and $\mathrm{pr}_{2}\colon \mathcal{M}\to\mathcal{M}_{7}$ are the corresponding canonical projections and $\mathrm{Vol_{4}}$ denotes the Lorentzian volume form on $\mathcal{M}_{1,3}$. With the choice of metric and four-form given in equations \eqref{eq:11dmetricproduct} and \eqref{eq:Gproduct} we guarantee that, given a particular Lorentzian four-dimensional manifold $(\mathcal{M}_{1,3},\mathsf{g}_{1,3})$, the equations of motion that the bosonic fields of the theory (the full metric $\mathsf{g}$ and four-form $\mathsf{G}$) need to satisfy, can written as a set of differential equation exclusively on $\mathcal{M}_{7}$. This will be important in section \ref{sec:modulispaces} in order to properly define the moduli space of this kind of solutions. 

Although the set-up presented above is more general, we will be usually interested in compactifications on maximally symmetric spaces $(\mathcal{M}_{1,3},\mathsf{g}_{1,3})$, which thus  will be taken to be the $AdS$  or Minkowski space. The $dS$-metric is discarded since it cannot be obtained supersymmetrically.


\subsection{Seven-dimensional supersymmetric $G_2$ structures} 


Supersymmetric solutions of eleven-dimensional Supergravity have to satisfy the Killing spinor equation of the theory, given by

\begin{equation}
\label{eq:N1susy11}
D_{v}\epsilon\equiv \nabla_{v}\epsilon + \frac{1}{6} \left(\iota_{v} \mathsf{G}\right) \epsilon +\frac{1}{12} \left(v^{\flat}\wedge \mathsf{G}\right)\epsilon = 0\, , \qquad \forall\,\, v\in \mathfrak{X}\left(\mathcal{M}\right)\, ,
\end{equation} 

\noindent
where $\mathsf{S}\to\mathcal{M}$ the 32-dimensional real, symplectic spinor bundle over $\mathcal{M}$, $\epsilon\in\Gamma\left(\mathsf{S}\right)$ is a Majorana spinor, $v^{\flat}$ is the one-form associated to $v$ through the metric and the terms in parenthesis act on the spinor via the Clifford multiplication.

Given the decomposition \eqref{eq:productmanifold} and the choice of metric \eqref{eq:11dmetricproduct}  we have to decompose the $Spin(1,10)$ spinor $\epsilon$ in terms of $Spin(1,3)\times Spin(7)$ representations. Let us denote by $\Delta_{\mathbb{R}}$ the 32-dimensional symplectic real representation of $Spin(1,10)$, by $\Delta^{+}_{1,3}$ and $\Delta^{-}_{1,3}$ the positive and negative chirality complex Weyl representation of $\mathrm{Spin}(1,3)\simeq Sl(2,\mathbb{C})$, of complex dimension two, and by $\Delta^{\mathbb{R}}_{7}$ the real Majorana representation of $\mathrm{Spin}(7)$. The branching rule $Spin(1,10)\to Spin(1,3)\times Spin(7)$ is given by\footnote{Let $V$ be a complex representation of a Lie group $G$. If $c\colon V\to V$ is a real structure on $V$, namely an invariant, antilineal map such that $c^2=1$, then $V=\left[V\right]\otimes_{\mathbb{R}}\mathbb{C}$, where $\left[V\right]$ is a real representation of $G$ of real dimension equal to the complex dimension of $V$.}

\begin{equation}
\label{eq:branchingrule}
\Delta^{\mathbb{R}}_{1,10} = \left[ \Delta^{+}_{1,3}\oplus \Delta^{-}_{1,3}\right]\otimes \Delta^{\mathbb{R}}_{7}\, ,
\end{equation}

\noindent
and therefore

\begin{equation}
\label{eq:11dspinor}
\epsilon = \xi\otimes \eta\, , \qquad\xi \in \Gamma\left[ S^{+}_{1,3}\oplus S^{-}_{1,3}\right]\, , \qquad \eta\in\Gamma\left(S^{\mathbb{R}}_{7}\right) \, ,
\end{equation}

\noindent
where $S^{\pm}_{1,3}$ is the positive (negative) chirality spin bundle over $\mathcal{M}_{1,3}$ and $S^{\mathbb{R}}_{7}$ is the real spin bundle over $\mathcal{M}_{7}$.

For each linearly independent eleven-dimensional spinor of the form (\ref{eq:11dspinor}) satisfying the Killing spinor equation (\ref{eq:N1susy11}), there is one out of 32 supersymmetries preserved. Therefore $\nu = \frac{1}{32}$ supersymmetry requires the existence of a single Majorana $Spin(7)$ spinor in the internal space $\mathcal{M}_{7}$\footnote{We could have also considered, as it is standard in the literature, the branching rule of the complexification of the supersymmetry spinor. Then, we would have ended up with two Majorana spinors on $\mathcal{M}_{7}$. Notice however that in that case the solutions are generically $\nu = \frac{1}{16}$ supersymmetric.}. 

The existence of a global spinor implies that the structure group of the spin bundle $S^{\mathbb{R}}_{7}$ is reduced from $Spin(7)$ to $G_2$.
 This structure can alternatively be defined by a positve 3-form $\varphi$\footnote{\label{foot:positive}A 3-form $\varphi$ is said to be positive \cite{Joyce2007} if at every point it can be written as $\varphi=dx^{123} + dx^{145} + dx^{167} + dx^{246} - dx^{257} - dx^{347} - dx^{356}$. Given a Majorana spinor $\eta$ and a metric $g_7$, the form $\varphi=\eta^T \gamma_{mnp} \eta \, dx^{mnp}$ is positive.}, which additionally
 induces a Riemannian metric $g_{7}$ that can explicitly be written as follows

\begin{equation}
\mathcal{V}_{7} g_{7}(v,w) = \frac{1}{3!}\iota_{v}\varphi\wedge\iota_{w}\varphi\wedge\varphi\, ,\qquad v, w \in \mathfrak{X}\left(\mathcal{M}_{7}\right)\, ,
\end{equation}

\noindent
where $\mathcal{V}_{7}$ is the Riemannian volume form of $\mathcal{M}_{7}$.

The supersymmetry conditions coming from the Killing spinor equation (\ref{eq:N1susy11}) can be translated into a set of differential conditions on $\varphi$.  In the absence of fluxes, and for a Minkowski four-dimensional vacuum these amount to $d \varphi=0, d*\varphi=0$, which implies that the structure is torsion-free, or equivalenetly that the manifold has $G_2$ holonomy. In the presence of fluxes the structure is not torsion-free anymore. It is useful to decompose the torsion classes into $G_2$ representations, namely 

\begin{equation}
\label{G2torsionclasses}
d\varphi=\tau_0 * \varphi + 3\,  \tau_1 \wedge \varphi+ * \tau_3 \ , \qquad d* \varphi=4 \,  \tau_1 \wedge * \varphi+\tau_2 \wedge  \varphi \, ,
\end{equation}

\noindent
where $\tau_{0,1,2,3}$ are a 0,1,2 and 3-form, respectively in the ${\bf 1},{\bf 7}, {\bf 14}$ and $ {\bf 27}$ representations of $G_2$. Supersymmetry relates these torsion classes to the different representations of the 4-form flux $G$. 
For a Minkowski vacuum, and particular decomposition of the supersymmetry spinor,  these relations can be found in \cite{Kaste:2003zd}.

 In this paper we are interested in the set of all supersymmetric solutions of the form given in \eqref{eq:11dmetricproduct} and \eqref{eq:Gproduct} as a space by itself; namely, we are interested in the moduli space of supersymmetric compactifications to four-dimensional $AdS$ or Minkowski space. Due to the presence of a non-trivial warp factor $e^{2A}$, a supersymmetric solution cannot be fully characterized by a set of first-order differential equations. Indeed, to obtain a supersymmetric solution to the full set of equations of motion of eleven-dimensional Supergravity one has to impose, on top of the supersymmetry conditions, the equation of motion for the flux. This is a second-order differential equation for $e^{2A}$. Hence, the moduli space of supersymmetric solutions to eleven-dimensional Supergravity cannot be characterized in terms of first-order differential equations. Therefore, we will consider instead the moduli space of supersymmetric configurations, which depends exclusively on the Killing spinor equation \ref{eq:N1susy11}, which is a first-order partial differential equation. We define then the concept of supersymmetric $G_{2}$-structure, which  will be used in section \ref{sec:modulispaces}, in the following way

\begin{definition}
\label{def:susyG2}
Let $(\mathcal{M}_{1,3},\mathsf{g}_{1,3})$ be a particular Lorentzian, four-dimensional, oriented and spin manifold and let $\mathcal{M}_{7}$ be a fixed seven-dimensional, oriented, spin manifold. A supersymmetric $G_{2}$-structure relative to $(\mathcal{M}_{1,3},\mathsf{g}_{1,3})$ is a quadruplet $\left(\varphi,G,e^{2A},\mu\right)$, where $\varphi\in\Omega_{+}\left(\mathcal{M}_{7}\right)$ is a positive three-form, $G\in\Omega_{cl}\left(\mathcal{M}_{7}\right)$ is a closed four-form, $e^{2A} \in C^{\infty}\left(\mathcal{M}_{7}\right)$ is a function and $\mu\in \mathbb{R}$, such that combined with the four-dimensional data in the form given by \eqref{eq:11dmetricproduct}, \eqref{eq:Gproduct}, \eqref{eq:11dspinor}, it obeys the Killing spinor equation \eqref{eq:N1susy11}\footnote{The Killing spinor equation \eqref{eq:N1susy11} is written in terms of the spinor $\eta$ instead of the 3-form $\varphi$, but from it one can obtain equations for $\varphi$ in terms of the components of $G$ and the warp factor (see (3.25) in reference \cite{Kaste:2003zd}).}.
\end{definition} 

\begin{remark}
Notice that depending on the choice of four-dimensional manifold $(\mathcal{M}_{1,3},\mathsf{g}_{1,3})$ the set of corresponding supersymmetric $G_{2}$-structures may be empty, for example if $(\mathcal{M}_{1,3},\mathsf{g}_{1,3})$ is $dS$-space, which is known to be non-supersymmetric. When talking about the set of supersymmetric $G_{2}$-structures we will always omit the space they are relative to, assuming that it is one such that the set is non-empty. 
\end{remark}



\subsection{Eight-dimensional intermediate manifolds}


Using the tools of generalized geometry, 
compactifications of M-theory down to four-dimensions with  $\mathcal{N}=1$ supersymmetry were geometrically characterised in \cite{Pacheco:2008ps} in terms of the space of sections of a real, 912-rank, vector bundle over $\mathcal{M}_{7}$

\begin{equation}
\label{eq:EbundleER}
\mathcal{E}\to\mathcal{M}_{7}\, ,
\end{equation}

\noindent
with structure group $E_{7(7)}\times\mathbb{R}^{+}$,  where $E_{7(7)}$ stands for the maximally non-compact real form of $E_{7}$ acting in the ${\bf 912}$ representation\footnote{Here ${\bf 912}$ denotes the real representation induced on $E_{7(7)}$ by the {\bf 912} of the exceptional complex Lie group $E_{7} = E_{7(7)}\otimes\mathbb{C}$. The ${\bf 912}$ representation of $E_{7}$ can be characterized as a particular subspace of $V_{56}\otimes \mathfrak{e}_{7}$, see appendix B in reference \cite{Pacheco:2008ps}.} and $\mathbb{R}^{+}$ represents conformal rescalings. 

 Let us denote by $N_{{\bf 912}}$ the vector space corresponding to the $\bf{912}$ representation of $E_{7(7)}\times\mathbb{R}^{+}$. In terms of $Sl(8,\mathbb{R})$ representations, it decomposes according to\footnote{\label{foot:E7subgroups}The real group $E_{7(7)}$ contains two important subgroups, namely $Sl(8,\mathbb{R})$ and $SU(8)/\mathbb{Z}_{2}$. The ${\bf 912}$ decomposes 
\begin{eqnarray}
{\bf 912} &=& \left[ {\bf 36} +{\bf 420} + {\bf \overline{36}} + {\bf \overline{420}}\right]\, ,\qquad SU(8)/\mathbb{Z}_{2}\subset E_{7(7)}\, , \nonumber\\
{\bf 912} &=&  {\bf 36} +{\bf 420} + {\bf 36}^{\prime} + {\bf 420}^{\prime}\, ,\qquad Sl(8,\mathbb{R})\subset E_{7(7)}\, , \nonumber
\end{eqnarray}

\noindent
where $\left[\cdot\right]$ denotes the corresponding induced real representation.}

\begin{eqnarray}
\label{eq:N912decomposition}
N_{912} & \simeq & S^2 V^{\ast} \oplus\left(\Lambda^3 V^{\ast} \otimes V \right)_{0} \oplus S^2V \oplus\left(\Lambda^3 V\otimes V^{\ast}\right)_{0}\, ,
\end{eqnarray}

\noindent
where $V$ is an eight-dimensional real vector space where $Sl(8,\mathbb{R})$ acts in the vector representation, $S$ stands for symmetric and the subscript 0 denotes traceless. 
%
%
%
%
We can apply the decomposition \eqref{eq:N912decomposition} to the bundle $\mathcal{E}$ \eqref{eq:EbundleER}, since at every point $p\in\mathcal{M}_{7}$, $\mathcal{E}_{p}$ is the {\bf 912} irreducible representation of $E_{7(7)}$. Hence we can write (\ref{eq:N912decomposition}) for $\mathcal{E}$
%
%
where now $V$ would denote a rank-eight real vector bundle over $\mathcal{M}_{7}$

\begin{equation}
V\to\mathcal{M}_{7}\, ,
\end{equation}

\noindent
with structure group $Sl(8,\mathbb{R})\times\mathbb{R}^{+}$ acting in the {\bf 8} of $Sl(8,\mathbb{R})$. Therefore, every section $\phi\in\Gamma\left(\mathcal{E}\right)$ can be decomposed as in \eqref{eq:N912decomposition}
%
%
%

Given a spinor  $\eta\in\Delta^{\mathbb{R}}_{7}$, there exists a canonical element $\psi^{\eta} \in N_{{\bf 912}}$ which is stabilized by $SU(7)\subset E_{7(7)}$ \cite{Pacheco:2008ps,Grana:2012zn} and can be written as follows in terms of $SU(8)$ representations (see footnote \ref{foot:E7subgroups})
\begin{equation}
\psi^\eta=(\eta \, \eta,0,0,0) \ .
\end{equation}
%
%

\begin{definition}
An element $\psi \in N_{912}$ is said to be E-admissible (or in the orbit of $\psi^\eta$), if there is a transformation $g\in E_{7(7)}\times\mathbb{R}^{+}$ such that

\begin{equation}
g\cdot \psi = \psi^{\eta}\, .
\end{equation}
%
\end{definition}

\begin{definition}
A section $\psi \in \Gamma\left(\mathcal{E}\right)$  is said to be E-admissible if, for every $p\in\mathcal{M}$, $\psi_{p}\in N_{912}$ is an E-admissible element of $N_{912}$.
\end{definition}

\begin{cor}
Let $\mathcal{E}\to\mathcal{M}$ be a rank-912 vector bundle with structure group $E_{7(7)}\times\mathbb{R}^{+}$. If $\mathcal{E}$ is equipped with a E-admissible section $\psi\in \Gamma\left(\mathcal{E}\right)$, then it admits a topological reduction of the frame bundle from $E_{7(7)}\times\mathbb{R}^{+}$ to $SU(7)$.
\end{cor}

\noindent
Off-shell ${\mathcal N}=1$ supersymmetry  implies the existence of an E-admissible section $\psi\in \Gamma\left(\mathcal{E}\right)$, and thus a topological reduction on $\mathcal{E}$ from $E_{7(7)}\times\mathbb{R}^{+}$ to $SU(7)$ and vice-versa, such a topological reduction implies the existence of a globally defined spinor on $\mathcal{M}_{7}$. It is then necessary to take a closer look at the space of E-admissible sections, since it corresponds to the space of topological reductions on $\mathcal{E}$ implied by supersymmetry. Notice that the space of E-admisible elements $N^{E}_{{\bf 912}}\subset N_{{\bf 912}}$ is transitive under the action of $E_{7(7)}\times\mathbb{R}^{+}$, and since every E-admissible element is stabilized by $SU(7)$ we obtain 

\begin{equation}
N^{E}_{{\bf 912}}\simeq \frac{E_{7(7)}}{SU(7)}\times\mathbb{R}^{+}\, .
\end{equation}

\noindent
Hence, the space of E-admissible elements $N^{E}_{{\bf 912}}$ is an 85-dimensional real homogeneous manifold times $\mathbb{R}^{+}$, and it was conjectured to be a K\"ahler-Hodge manifold\footnote{To the best of our knowledge this has not been proven yet} in reference \cite{Pacheco:2008ps} in order to make contact with the geometry of the non-linear sigma-model of $\mathcal{N}=1$ Supergravity. Let us define then the fibre bundle

\begin{equation}
\mathcal{E}_{E} \to \mathcal{M}_{7}\, ,
\end{equation}

\noindent
with fibre $\mathcal{E}_{E\, p}\simeq  \frac{E_{7(7)}}{SU(7)}\times\mathbb{R}^{+}\, , \, \, p\in\mathcal{M}_{7}$.
The set $\tilde{\mathcal{X}}$ of all E-admissible structures on $\mathcal{E}$ can be thus written as the space of sections of $\mathcal{E}_{E}$

\begin{equation}
\tilde{\mathcal{X}} = \left\{\psi\,\, |\,\, \psi\in \Gamma\left(\mathcal{E}_{E}\right)\right\}\, .
\end{equation}

\noindent
In order to properly define the moduli space of supersymmetric compactifications as a subset of $\tilde{\mathcal{X}}$ there are two missing pieces. These are the differential conditions implied on $\psi\in\Gamma\left(\mathcal{E}_{E}\right)$ by supersymmetry, which come from the  Killing spinor equation \eqref{eq:N1susy11}, and the quotient by the appropriate equivalence relation. If we had them we could readily characterize the moduli space of $\mathcal{N}=1$ compactifications to four dimensions as follows

\begin{equation}
\mathcal{X} = \left\{\psi\,\, |\,\, \tilde{d}\psi = 0\, , \,\, \psi\in \Gamma\left(\mathcal{E}_{E}\right)\right\}/{\cal D}\, .
\end{equation}

\noindent
where $\tilde{d}$ is an appropriate coboundary operator and ${\cal D}$ are generalized diffeomorphisms, combining diffeomorphisms and gauge transformations. However, given on one hand that we do not know the differential conditions, and on the other that there are not yet available any mathematical results on deformations on these structures, we will 
pursue here a different path, more straightforward in some sense, though restricted to particular set of supersymmetric backgrounds.  

We have seen that off-shell ${\cal N}=1$ supersymmetry is equivalent to the existence of an E-admissible section, which at every point $p\in\mathcal{M}_{7}$ can be written as as follows in terms of $SL(8,{\mathbb R})$ representations  (cf. Eq. (\ref{eq:N912decomposition})) \cite{Grana:2012zn}

\begin{equation}
\label{eq:canonicaldecompositionII}
\psi = g_{8}\oplus \phi \oplus g^{-1}_{8}\oplus \tilde{\phi}\, ,
\end{equation}

\noindent
where $g_{8}$ is a Riemannian metric on $V\to\mathcal{M}_{7}$ and $\phi$, $\tilde{\phi}$ are appropriate elements in the second and fourth components of (\ref{eq:N912decomposition}). Therefore $(g_{8},\phi,\tilde{\phi})$ carry all the information of the E-admissible section of $\mathcal{E}$.
However, $\mathcal{E}$ is an extrinsic bundle over $\mathcal{M}_{7}$, while it would be desirable to have a description given in terms of intrinsic bundles, namely tensor bundles of a given manifold, since they are easier to handle. 

Given that the exceptional bundle $\mathcal{E}$ can be written as the Whitney sum of rank-eight real vector bundles over the seven dimensional base $\mathcal{M}_{7}$,  we propose an eight-dimensional manifold $\mathcal{M}_{8}$ such that it contains in its tensor bundles all the information that is present in the set of E-admissible sections of $\mathcal{E}$. This way, the role of the rank-eight vector bundle $V\to\mathcal{M}_{7}$ is played by the tangent bundle $T\mathcal{M}_{8}$ of $\mathcal{M}_{8}$, which then must be equipped with the corresponding sections to contain a decomposition of  the ${\bf 912}$ representation of $E_{7(7)}\times \mathbb{R}^{+}$ in terms of tensors over $\mathcal{M}_{8}$. This motivates the following definition.

\begin{definition}
\label{def:intermediatemanifolds}
Let $\mathcal{M}_{8}$ be an eight-dimensional, oriented, spin differentiable manifold. We say that $\mathcal{M}_{8}$ is an  intermediate manifold if it is equipped with the following data
\begin{itemize}

\item A Riemannian metric $g\in\Gamma\left( S^2 T^{\ast}\mathcal{M}_{8}\right)$.

\item A globally defined section $\phi\in\Gamma\left(\Lambda^{3} T^{\ast}\mathcal{M}_{8}\otimes T\mathcal{M}_{8}\right)_{0}$.

\item A global section globally defined $\tilde{\phi}\in\Gamma\left(\Lambda^{3} T\mathcal{M}_{8}\otimes T^{\ast}\mathcal{M}_{8}\right)_{0}$.

\end{itemize}

\noindent
We will say then that $\mathfrak{S} = \left(g,\phi, \tilde{\phi}\right)$ is an intermediate structure on $\mathcal{M}_{8}$.
\end{definition}

\noindent
Therefore, we relate to every compactification manifold $\mathcal{M}_{7}$ an eight-dimensional manifold $\mathcal{M}_{8}$ whose physical role remains to be found. 

Since an intermediate manifold $\left(\mathcal{M}_{8},\mathfrak{S}\right)$ is equipped with the same tensors than those implied by an E-admissible section on $\mathcal{E}$, we expect a very close relation between the space $\tilde{\mathcal{X}}$ of E-admissible sections and the space of intermediate structures on $\mathcal{M}_{8}$, which we will denote by $\tilde{\mathfrak{H}}$. Now, the difference between the description in terms of $\tilde{\mathcal{X}}$ and the description in terms of $\tilde{\mathfrak{H}}$ is that on one hand in the latter we expect  the \emph{differential conditions} that are implied by supersymmetry to have a simpler form and on the other it is a more manageable space, since it consists of sections of tensor bundles, instead of sections of an extrinsic exceptional bundle. However, appropriately lifting  the set of forms, together with their differential conditions, to the eight-dimensional manifold $(\mathcal{M}_{8},\mathfrak{S})$ is a tricky step, since there is not an obvious or unique way of doing it. As explained in \cite{Grana:2014rpa}, in the simplest case where the intermediate manifold $\left(\mathcal{M}_{8},\mathfrak{S}\right)$ is such that

\begin{equation}
\phi = \phi_{3}\otimes v\, ,  \qquad \iota_{v}\phi_{3} = 0\, ,
\end{equation}

\noindent
where $\phi_3 \in \Omega^3 ({\mathcal M}_8)$ and $v\in\mathfrak{X}\left(\mathcal{M}_{8}\right)$. In reference \cite{Grana:2014rpa} we showed that it is natural to consider $\mathcal{M}_{8}$ either as a rank-one vector bundle over $\mathcal{M}_{7}$ or as having a codimension-one foliation in terms of leaves such that at least one of them is diffeomorphic to $\mathcal{M}_{7}$. Furthermore, if $\phi_3$ is the unique lift of a positive positive three-form on $\mathcal{M}_{7}$, the following four-form \cite{Grana:2014rpa}

\begin{equation}
\Omega = v^{\flat}\wedge  \phi_3+ \ast\left(v^{\flat}\wedge  \phi_3 \right)\, ,
\end{equation}

\noindent 
is admissible\footnote{\label{foot:admissible}A four-form is said to be admissible if it can be written at every point  as 
$\Omega_{0} = dx^{1234} + dx^{1256} + dx^{1278} + dx^{1357} - dx^{1368} - dx^{1458}  - dx^{1467}
 - dx^{2358} - dx^{2367}- dx^{2457} +  dx^{2468} + dx^{3456} + dx^{3478} + dx^{5678}.$}, or in other words it defines a $Spin(7)$-structure on $\mathcal{M}_{8}$. 

We now turn to the differential conditions on the intermediate manifold imposed by supersymmetry.


\subsection{Spin(7) holonomy on ${\mathcal M}_8$}
\label{sec:Spin7Holonomy}


We will not attempt to uplift the differential conditions in their full generality, but we note that remarkably enough, one can impose a very natural and useful condition on $\mathcal{M}_{8}$ that basically fixes the way in which $\mathcal{M}_{8}$ is constructed from $\mathcal{M}_{7}$, and provides $\mathcal{M}_{8}$ with a nice geometrical structure with more well-known mathematical properties. This condition is the requirement that $(\mathcal{M}_{8},\mathfrak{S})$ has at least $Spin(7)$-holonomy. Restricting to a given class of $G_2$ structures, $\mathcal{M}_{8}$ can be built from $\mathcal{M}_{7}$ in such a way that for every restricted $G_{2}$-structure on $\mathcal{M}_{7}$ there is a unique torsion-free $Spin(7)$-structure on $\mathcal{M}_{8}$. As we will see in section \ref{sec:differentialconditions}, one can describe this way a particular set of torsion-full $G_2$ structures, namely those with constant $\tau_0$ and harmonic $\tau_2$ torsion classes (see Eq. (\ref{G2torsionclasses})).

There are basically two reasons to impose $Spin(7)$-holonomy on $(\mathcal{M}_{8},\mathfrak{S})$, namely

\begin{itemize}

\item Due to the presence of fluxes, $\mathcal{M}_{7}$ is equipped with a $G_{2}$-structure which is not torsion free. The moduli space of such $G_{2}$-structures is very poorly understood, and thus if we \emph{relate} $\mathcal{M}_{7}$ to an eight-dimensional manifold $\mathcal{M}_{8}$ in such a way that the supersymmetric $G_{2}$-structure induces a torsion-free $Spin(7)$-structure, we are \emph{geometrizing} the fluxes and at the same time precisely relating $\mathcal{M}_{7}$ with a manifold $\mathcal{M}_{8}$ whose moduli space\footnote{Meaning the moduli space of torsion-free $Spin(7)$ structures on $\mathcal{M}_{8}$.} is much more under control. Notice that this is an interesting problem by itself, aside from the study of intermediate manifolds. A word of caution is needed here though: constructing a compact $Spin(7)$-holonomy manifold $\mathcal{M}_{8}$ \emph{from} $\mathcal{M}_{7}$ is an extremely difficult task, one of the reasons being that a compact $Spin(7)$-holonomy manifold has no isometries. On the other hand, constructing a non-compact $\mathcal{M}_{8}$ is doable, but then in that case the moduli space of non-compact $Spin(7)$-manifolds is not as well understood as in the compact counterpart, although some results are known for cases with controlled asymptotics\footnote{We thank Dominic Joyce for a clarification about this point.}.  

\item The eight-dimensional manifold $\mathcal{M}_{8}$ is constructed from the physical seven-dimensional one $\mathcal{M}_{7}$ and thus it is natural to study if it has some physical significance beyond being an auxiliary tool. Since we are considering compactifications to four-dimensions and $\mathcal{M}_{8}$ is eight-dimensional, it is natural to study the possibility of $\mathcal{M}_{8}$ being an admissible compactification background for F-theory. Therefore it is reasonable to impose at least $Spin(7)$-holonomy, since a large class of F-theory internal spaces have special holonomy. Of course, in addition $(\mathcal{M}_{8},\mathfrak{S})$ must be elliptically fibered. Indeed it was shown in \cite{Grana:2014rpa} that an intermediate structure $\mathfrak{S}$ completely characterizes an elliptic fibration on $\mathcal{M}_{8}$, pointing out to a possible role for $\mathcal{M}_{8}$ as an F-theory internal space. The question is then what would be the relation between M-theory compactified on $\mathcal{M}_{7}$ and F-theory compactified on $\mathcal{M}_{8}$, and this could be the source of new M/F-theory dualities. There are in the literature examples of this kind of constructions arising from M/String-theory dualities. For instance, in reference \cite{Kaste:2002xs} it was shown that a configuration of D6-branes wrapping a special Lagrangian cycle of a non-compact Calabi-Yau manifold such that the internal string frame metric is K\"ahler, is dual to a purely geometrical background in eleven dimensions with an internal metric of $G_{2}$-holonomy. Other examples of torsion-free structures dual to configuration with fluxes, in other words, non-zero torsion, can be found in references \cite{2002math......2282C,2004CMaPh.246...43A,Santillan:2006rr,Giribet:2006xn}, this time from six-dimensions to a $G_{2}$-holonomy. For further examples and results in the $Spin(7)$ and $G_{2}$ cases see reference \cite{Salur:2008pi,delaOssa:2014lma}.
 
\end{itemize}


\section{Examples}
\label{sec:differentialconditions}


In this section we construct examples of intermediate manifolds with $Spin(7)$ holonomy starting from supersymmetric $G_2$ structures in seven dimensions.
%


\subsection{$Spin(7)$-manifolds from trivial products of $G_{2}$-structure manifolds and $\mathbb{R}$}
\label{sec:Spin7cone}


Let $\mathcal{M}_{7}$ be a seven-dimensional oriented manifold equipped with a positive three-form $\varphi$. Let $\mathcal{M}_{8} = \mathbb{R}^{+}\times\mathcal{M}_{7}$ be the topologically trivial product of $\mathbb{R}^{+}$ and $\mathcal{M}_{7}$, which we take to be oriented with volume form $\mathcal{V}_{8} = dt\wedge\mathcal{V}_{7}$. We define the following four-form $\Omega$ on $\mathcal{M}_{8}$

\begin{equation}
\label{eq:Omegafromvarphi}
\Omega = f_{1}\, dt\wedge\varphi + f_{2}\,\ast_{7}\varphi\, ,
\end{equation}

\noindent
where $t$ is the natural coordinate on $\mathbb{R}^{+}$ and $f_{1}, f_{2}\in C^{\infty}\left(\mathcal{M}_{8}\right)$ are positive functions on $\mathcal{M}_{8}$. 

\begin{lemma}
\label{lemma:admissibleRM7}
The four-form $\Omega$ as in \eqref{eq:Omegafromvarphi} is an admissible four-form in $\mathcal{M}_{8}$ if and only if $\varphi$ is a positive form on $\mathcal{M}_{7}$. 
\end{lemma}

\noindent
The four-form $\Omega$ defines therefore a unique topological reduction on $F\left(\mathcal{M}_{8}\right)$ from $Gl_{+}\left(8,\mathbb{R}\right)$ to $Spin(7)$, which however is redundant since it can be further reduced to $G_{2}$. Notice however that the holonomy of $\mathcal{M}_{8}$ is not $G_{2}$, unless, as we will see in a moment, $f_{1} = f_{2} =1$, but in that case $\mathcal{M}_{8}$ is a reducible Riemannian manifold and thus not covered by Berger's list of possible holonomy groups. Furthermore, the admissible four-form $\Omega$ induces a Riemannian metric $g_{8}$ on $\mathcal{M}_{8}$\footnote{Since the expression for $g_{8}$ in terms of $\Omega$ is not particularly illuminating and will not be needed, we do not write it here. It can be found in \cite{2003math......1218K}, theorem 4.3.5.}. Taking $f _1 = f_2 = 1$ the following result can be proven.

\begin{prop}{\bf [Proposition 13.1.3 \cite{Joyce2007}]}
Let $\left(\mathcal{M}_{8},\Omega\right)$ be as above and let us take $f _1 = f_2 = 1$. Then, $\left(\Omega,g_{8}\right)$ is a torsion-free $Spin(7)$ structure on $\mathcal{M}_{8}$ if and only if $\left(\varphi,g_{7}\right)$ is a torsion-free $G_{2}$ structure. The associated metric induced by $\Omega$ is $g_{8} = dt^2\times g_{7}$.
\end{prop}

\noindent
This is the simplest case of the construction of an intermediate manifold $\left(\mathcal{M}_{8},\mathfrak{S}\right)$ from a seven-dimensional manifold $\mathcal{M}_{7}$, in this case of $G_{2}$ holonomy. The intermediate structure $\mathfrak{S} = \left(g,\phi,\tilde{\phi}\right)$ is given by

\begin{equation} \label{interm}
g = g_{8}\, , \qquad \phi = \varphi_{3}\otimes v\, , \qquad v = \partial_{t}\, ,
\end{equation}

\noindent
and $\tilde{\phi}$ given by the dual of $\phi$ by the corresponding musical isomorphisms. However, this construction does not illustrate the \emph{geometrization} of fluxes since the $G_{2}$ structure $\left(\varphi,g_{7}\right)$ is already torsion-free. This situation can be modified by considering non-constant functions $f_{1}$ and $f_{2}$. Let us first consider the cone-like construction \cite{Bar}. 

\begin{prop} 
\label{prop:Spin7cone}
Let $\left(\mathcal{M}_{8},\Omega\right)$ be as above and let us take $f_1, f_2 \in C^{\infty}\left(\mathbb{R}^{+}\right)$ be positive functions on $\mathbb{R}^{+}$. Then, $\left(\Omega,g_{8}\right)$ is a torsion-free $Spin(7)$-structure on $\mathcal{M}_{8}$ different from the product one  if and only if $\left(\varphi,g_{7}\right)$ is a nearly-parallel (also called weak) $G_{2}$-structure, namely

\begin{equation}
d\varphi = \lambda\ast_{7}\varphi\, , \qquad d\ast_{7}\varphi = 0\, , \qquad \lambda\in\mathbb{R}^{+}\, ,
\end{equation}

\noindent
and $f_{1} = 1\, , f_{2} =\lambda t$. 
\end{prop}

The proof is straightforward: since $f_{1}$ can  always be eliminated by a change of coordinates on $t$ given by $\frac{d\tilde{t}}{dt} = f_{1}(t)$, we will assume that $f_{1} = 1$ without relabelling $t$. The condition $d\Omega = 0$ then translates into 
\begin{equation}
-dt\wedge d\varphi + \partial_{t} f_{2} dt\wedge \ast_{7}\varphi + f_{2} d\ast_{7}\varphi = 0\, ,
\end{equation}

\noindent
and thus 

\begin{equation}
d\ast_{7}\varphi = 0\, ,\quad  d \varphi= \partial_{t} f_{2}  \ast_7 \varphi
\end{equation}

\noindent
which implies 

\begin{equation}
\tau_1= 0 \ , \quad \tau_2=0, \quad \tau_{3} = 0\, , \quad \qquad f_{2} = \tau_{0}\, t\, , \qquad \tau_{0}\in\mathbb{R}^{+}\, .
\end{equation}
where $\tau_i$ are the torsion classes defined in (\ref{G2torsionclasses}).
\noindent
Again $\left(\mathcal{M}_{8},\Omega\right)$, as in the preposition above can be easily embedded in an intermediate structure $\mathfrak{G}$ 
as in (\ref{interm}). However in this case, since $\varphi$ defines a nearly $G_{2}$-structure on $\mathcal{M}_{7}$, we obtain a torsion-free $Spin(7)$-structure from a $G_{2}$-structure which is not torsion-free. This is therefore the simplest non-trivial example of intermediate manifold, which is constructed as a cone over a seven-dimensional compact nearly $G_{2}$-manifold (namely $\tau_{0}$ is the only non-vanishing torsion-class), in such a way that we obtain an a torsion free $Spin(7)$-structure on the eight-dimensional manifold $\mathcal{M}_{8}$. The relation between nearly $G_{2}$-structures on $\mathcal{M}_{7}$ and torsion-free structures on $\mathcal{M}_{8}$ is in fact one-to-one.

Unfortunately, this is the only torsion class that can be geometrized using (\ref{eq:Omegafromvarphi}). Indeed, if we take $f_{1}$ and $f_{2}$ as generic real positive functions on $\mathcal{M}_{8}$ the following result holds.
\begin{prop}
Let $\left(\mathcal{M}_{8},\Omega\right)$ satisfying the conditions of lemma \ref{lemma:admissibleRM7}, with $f_{1}, f_{2} \in C^{\infty}\left(\mathcal{M}_{8}\right)$. Then, $\left(\Omega,g_{8}\right)$ is a torsion-free $Spin(7)$-structure on $\mathcal{M}_{8}$ different from the product one if and only if $\left(\varphi,g_{7}\right)$ is a nearly parallel $G_{2}$-structure and $f_{1} =1\, , f_{2} = \lambda t$. 
\end{prop}

\begin{proof}
The condition $d\Omega = 0$ splits in two independent requirements, namely

\begin{equation}
\label{eq:conditionsproduct}
d_7 f_{1}\wedge dt\wedge\varphi - f_{1}\, dt\wedge d_{7}\varphi + \partial_{t}f_{2}\, dt\wedge \ast_{7}\varphi = 0\, ,\qquad d_{7}\left( f_{2}\,\ast_{7}\varphi\right) = 0\, ,
\end{equation}

\noindent
where the subscript $7$ denotes objects on $\mathcal{M}_{7}$. The second equation can be used to rescale $\varphi$ so that we eliminate the dependence of $f_{2}$ in $\mathcal{M}_{7}$, so we will assume that $f_{2}\in C^{\infty}\left(\mathbb{R}\right)$ and thus

\begin{equation}
d_{7}\ast_{7}\varphi = 0\, .
\end{equation}

\noindent
which implies $\tau_{1} = \tau_{2} = 0$ and thus

\begin{equation}
\label{eq:phitorsionRM7}
d\varphi = \tau_{0}\ast_{7}\varphi +\ast_{7}\tau_{3}\, ,
\end{equation}

\noindent
where $\tau_{0}\in C^{\infty}(\mathcal{M}_{7})$ and $\tau_{3}\in\Omega^{3}_{27}(\mathcal{M})$. Using this in the first equation of \eqref{eq:conditionsproduct} we obtain

\begin{equation}
\label{eq:torsionconditionsRM7}
f_{1} = \frac{\lambda}{\tau_{0}}\, , \qquad f_{2} = \lambda t\, , \qquad \tau_{3} = -\ast_{7}\left( d_{7}\log\, \tau_0\wedge \varphi\right)\, .
\end{equation} 

\noindent
However, $\tau_{3}\in\Omega^{3}_{27}\left(\mathcal{M}_{7}\right)$ if and only if $\tau_{3} = 0$ and thus $\tau_{0}$ is constant and we arrive at the case discussed before.

\end{proof}

\noindent
Therefore, the previous proposition proves that if we impose $Spin(7)$-holonomy on an eight-dimensional manifold $\mathcal{M}_{8}$ of the form $\mathcal{M}_{8} = \mathbb{R}\times\mathcal{M}_{7}$ then the only possibility is a cone over a nearly-$G_2$ manifold. 

Nearly $G_2$-structures appear in Freund-Rubin-type of ${\mathcal N}=1$ compactifications to $AdS_4$-space \cite{Freund:1980xh,Duff:1986hr,Acharya:2003ii}, where the radius of $AdS$ is given by the four-form flux $\mu$, which is proportional to $\tau_0$. In particular, the space-time is assumed to be of the form

\begin{equation}
\mathcal{M} = AdS_{4}\times\mathcal{M}_{7}\, , 
\end{equation}

\noindent
equipped with the product metric 

\begin{equation}
g = g_{AdS_{4}} + g_{7}\, ,
\end{equation}

\noindent
and four-form flux $G$ given by the volume form of the $AdS$-space

\begin{equation}
G = \mu {\rm Vol}_{AdS_{4}}\, .
\end{equation}

\noindent
Under this provisos, it can be shown that $\mathcal{M}_{7}$ must be a nearly-$G_{2}$ manifold and thus the cone construction geometrizes the flux present in this class of compactification backgrounds.

We want to consider now more general intermediate manifolds, or, from another point of view, more general embeddings of supersymmetric  $G_{2}$-structures into torsion-free $Spin(7)$-structures. In the following section we are going to consider the natural next step, which is the case of $\mathcal{M}_{8}$ being an $S^{1}$ principal bundle. 

\subsection{$Spin(7)$-manifolds as $S^{1}$-bundles}
\label{sec:Spin7bundle}


Let $\mathcal{M}_{8}\xrightarrow{\pi}\mathcal{M}_{8}/S^{1}$ be an $S^{1}$-bundle over ${\mathcal M}_7$, equipped with an $S^{1}$-invariant $Spin(7)$-structure $\Omega$, which is a non-trivial instance of intermediate manifold. Then, as proven in reference \cite{Grana:2014rpa}, there is an induced $G_{2}$-structure on the base $\mathcal{M}_{8}/S^{1}$ and conversely, a $G_{2}$-structure $\varphi$ on $\mathcal{M}_{8}/S^{1}$ induces a $Spin(7)$-structure on $\mathcal{M}_{8}$. Now we are going to show that this correspondence is one-to-one, providing thus a non-trivial example of the map $\mathfrak{I}_{Top}$ to be introduced in section \ref{sec:modulispaces}. First we need the following proposition.

\begin{prop}{\bf [Proposition 3.12 \cite{Grana:2014rpa}]}
\label{prop:bijectivemaplinear}
Let $V$ denote an eight-dimensional oriented vector space equipped with a fixed metric $g$. Fix $v\in V$ of unit norm. Let $i^{\ast}g$ denote the restricted metric on the seven-dimensional subspace $H\equiv v^{\perp}$ endowed with the orientation induced by that on $V$ and by $v$. Then there is a bijection

\begin{eqnarray}
\mathfrak{T}\colon \left\{\Lambda^{4}_{a}\left(V\right), g\right\} &\to & \left\{\Lambda^{3}_{+}\left(H\right), i^{\ast}g\right\}\nonumber\\
\Omega &\mapsto & \phi\equiv\iota_{v}\Omega
\end{eqnarray}

\noindent
 from the space of all admissible four-forms $\Lambda^{4}_{a}\left(V\right)$ on $V$ inducing the metric $g$ to the space of all positive three-forms $\Lambda^{3}_{+}\left(H\right)$ on $H$ inducing $i^{\ast}g$. The inverse map is given by

\begin{eqnarray}
\mathfrak{T}^{-1}\colon \left\{\Lambda^{3}_{+}\left(H\right), i^{\ast}g\right\} &\to &\left\{\Lambda^{4}_{a}\left(V\right), g\right\}\nonumber\\
\varphi &\mapsto & v^{\flat}\wedge\varphi + \ast\left(v^{\flat}\wedge\varphi\right) \, ,
\end{eqnarray}

\noindent
where $\varphi = \phi |_{H}$.
\end{prop}

\noindent
As a direct application of proposition \ref{prop:bijectivemaplinear} we have the following corollary.

\begin{cor}
Let $V$ denote an eight-dimensional oriented vector space equipped with an admissible four-form $\Omega\in \Lambda^{4}_{a}\left( V\right)$. Then, for every unit norm vector $v\in V$, $\Omega$ can be written as

\begin{equation}
\Omega = v^{\flat}\wedge\phi + \ast\left(v^{\flat}\wedge\phi\right)\, ,
\end{equation}

\noindent
where $\phi = \iota_{v}\Omega$.
\end{cor}

\noindent
The following theorem extends proposition \ref{prop:bijectivemaplinear} to the case of an eight-dimensional $S^{1}$-bundle equipped with an invariant $Spin(7)$-structure $\Omega$. 

\begin{thm}
\label{thm:bijectionM8}
Let $\mathcal{M}_{8}\xrightarrow{\pi} \mathcal{M}_{8}/S^{1}$ be an eight-dimensional oriented $S^{1}$-bundle equipped with a $S^{1}$-invariant metric $g_{8}$. Let $\mathcal{M}_{8}/S^{1}$ be the base space equipped with the unique metric $g_{7}$ induced from $g_{8}$ by means of the connection one-form $\theta = v^{\flat}$, where $v$ is the infinitesimal generator of the $S^{1}$-action, such that $\pi\colon\mathcal{M}_{8}\to\mathcal{M}_{8}/S^{1}$ is a Riemannian submersion. Then, there is a bijection

\begin{eqnarray}
\mathfrak{T}\colon \left\{\Omega^{4}_{a}\left( M_{8}\right), g_{8}\right\} &\to & \left\{\Omega^{3}_{+}\left(\mathcal{M}_{8}/S^{1}\right), g_{7}\right\}\nonumber\\
\Omega &\mapsto & \varphi
\end{eqnarray}

\noindent
from the space of all admissible four-forms $\Omega^{4}_{a}\left( M_{8}\right)$ inducing $g_{8}$ to the space of all positive three-forms $\Omega^{3}_{+}\left(\mathcal{M}_{8}/S^{1}\right)$ inducing $g_{7}$. Here $\varphi$ stands for the unique three-form on $\mathcal{M}_{8}/S^{1}$ such that $\pi^{\ast} \varphi = \phi \equiv \iota_{w}\Omega$, where we have defined $w = \frac{v}{\| v\|}$. The inverse map is given by

\begin{eqnarray}
\mathfrak{T}^{-1}\colon \left\{\Omega^{3}_{+}\left(\mathcal{M}_{8}/S^{1}\right), g_{7}\right\} &\to &\left\{\Omega^{4}_{a}\left( M_{8}\right), g_{8}\right\}\nonumber\\
\varphi &\mapsto & w^{\flat}\wedge\phi + \ast\left( w^{\flat}\wedge\phi\right)
\end{eqnarray}
\end{thm}

\begin{proof}
Let us denote by $\| v\|$ the norm of the vector $v$. It can  easily be seen that $\mathcal{L}_{v}\| v\| = 0$ and thus $\| v\|$ descends to a well defined function on $\mathcal{M}_{8}$, so we are in the conditions of theorems 5.27 and 5.29 in \cite{Grana:2014rpa}, from which we see that the map $\mathfrak{T}$ and its inverse $\mathfrak{T}^{-1}$ are well-defined. Let us see now that it is a bijection. Let us use the kernel of the connection $\theta = v^{\flat}$ to define a co-dimension one distribution $H \subset T\mathcal{M}_{8}$, which thus implies the splitting

\begin{equation}
T\mathcal{M}_{8} = H\oplus \left\{\lambda v\right\}\, , \qquad\lambda\in \mathbb{R}\, .
\end{equation}

\noindent
Therefore, at every point $p\in\mathcal{M}_{8}$ we are in the conditions of proposition \ref{prop:bijectivemaplinear}, and thus at every point there is a bijection $\mathfrak{T}^{\prime}_{p}$ between the set of admissible four-forms in $T_{p}\mathcal{M}_{8}$ inducing $g_{8}|_p$ and the set of three-forms on $H_{p}$ inducing $g_{8}|_{H}$. Now, $d\pi\colon H\to T\left(\mathcal{M}_{8}/S^{1}\right)$ is an isometry and then point-wise composing with $\mathfrak{T}^{\prime}$ we obtain $\mathfrak{T}$ and we conclude.
\end{proof}

\noindent
As a consequence of theorem \ref{thm:bijectionM8}, we obtain the following proposition.

\begin{prop}
\label{prop:S1Spinbundle}
Every eight-dimensional, oriented spin $S^{1}$-bundle $\mathcal{M}_{8}\xrightarrow{\pi} \mathcal{M}_{8}/S^{1}$ admits a topological $Spin(7)$-structure. This $Spin(7)$-structure can be chosen to be $S^{1}$-invariant.
\end{prop}

\begin{proof}
Let us pick a $S^{1}$-invariant metric $g_{8}$ on $\mathcal{M}_{8}$ and let $\theta = v^{\perp}$ be a connection with horizontal distribution $H\subset T\mathcal{M}_{8}$. Then we have

\begin{equation}
T\mathcal{M}_{8} = H\oplus V\, ,
\end{equation}

\noindent
where $V\to\mathcal{M}_{8}$ is a one-dimensional vector bundle. Since $v\in\Gamma\left(V\right)$ is a global section $V$ is oriented and spin, and thus $H\to\mathcal{M}_{8}$ is also an oriented spin vector bundle, namely $w_{1}(H) = w_{2}(H) = 0$. Now, $\pi\colon\mathcal{M}_{8}\to \mathcal{M}_{8}/S^{1}$ is a continuous  map and since the first and second Stiefel-Whitney classes are natural \cite{Spingeometry}, we obtain 

\begin{equation}
w_{1}\left(\pi^{\ast}T\left(\mathcal{M}_{8}/S^{1}\right)\right) = \pi^{\ast}w_{1}\left(T\left(\mathcal{M}_{8}/S^{1}\right)\right) \, , \qquad w_{2}\left(\pi^{\ast}T\left(\mathcal{M}_{8}/S^{1}\right)\right) = \pi^{\ast}w_{2}\left(T\left(\mathcal{M}_{8}/S^{1}\right)\right) \, ,
\end{equation}

\noindent
where $\pi^{\ast}T\left(\mathcal{M}_{8}/S^{1}\right)$ denotes de pull-back vector bundle of $T\left(\mathcal{M}_{8}/S^{1}\right)$ by $\pi$. But $\pi^{\ast}T\left(\mathcal{M}_{8}/S^{1}\right)\simeq H$ and thus 

\begin{equation}
w_{1}\left( H\right) = \pi^{\ast}w_{1}\left(T\left(\mathcal{M}_{8}/S^{1}\right)\right) = 0 \, , \qquad w_{2}\left( H\right) = \pi^{\ast}w_{2}\left(T\left(\mathcal{M}_{8}/S^{1}\right)\right)= 0\, ,
\end{equation}
which implies $w_{1}\left(T\left(\mathcal{M}_{8}/S^{1}\right)\right) = 0$ and $w_{2}\left(T\left(\mathcal{M}_{8}/S^{1}\right)\right) = 0$, namely $\mathcal{M}_{8}/S^{1}$ is a spin manifold, and hence  it is equipped with a $G_{2}$-structure $\varphi$. Applying theorem \ref{thm:bijectionM8} we obtain that $\mathfrak{J}^{-1}\left(\varphi\right)$ is an $S^{1}$-invariant admissible four-form in $\mathcal{M}_{8}$. Once we now that $\mathcal{M}_{8}/S^{1}$ is spin, lifting the corresponding $G_{2}$-structure using in [Theorem 5.27 \cite{Grana:2014rpa}] a codimension-one distribution $H$ perpendicular to $v$ but non-invariant results in a non-invariant $Spin(7)$-structure on $\mathcal{M}_{8}$.
\end{proof}

\begin{remark}
Theorem \ref{thm:bijectionM8} implies that if there are no admissible four-forms on $\mathcal{M}_{8}$ compatible with $g_{8}$, then the base space $\mathcal{M}_{8}/S^{1}$ cannot be equipped with a $G_{2}$-structure and therefore cannot be a spin manifold, which in turn implies than $\mathcal{M}_{8}$ is not a spin manifold, in agreement with proposition \ref{prop:S1Spinbundle}.
\end{remark}

\noindent
Theorem \ref{thm:bijectionM8} gives us the first explicit example of a map $\mathfrak{J}$ relating $Spin(7)$-structures into $G_{2}$-structures and vice-versa. Following the discussion of section \ref{sec:modulispaces} we have to study now the differential conditions on the induced $G_{2}$-structure on $\mathcal{M}_{8}/S^{1}$ implied by requiring a torsion-free $S^{1}$-invariant $Spin(7)$-structure on $\mathcal{M}_{8}$. This way we will see if it is possible to have a one-to-one map from $G_{2}$-structures on $\mathcal{M}_{8}/S^{1}$ with some non-zero torsion, to torsion-free $Spin(7)$-structures on $\mathcal{M}_{8}$. We will prove a proposition that completely characterizes the situation. 

\begin{prop}
\label{prop:M8holonomy}
Let $\mathcal{M}_{8}\xrightarrow{\pi} \mathcal{M}_{8}/S^{1}$ denote an eight-dimensional oriented $S^{1}$-bundle equipped with a $S^{1}$-invariant admissible four-form $\Omega\in\Omega^{4}_{a}\left(\mathcal{M}_{8}\right)$ and let $\varphi\in\Omega^{3}_{+}\left(\mathcal{M}_{8}/S^{1}\right)$ be the induced $G_{2}$-structure on $\mathcal{M}_{8}/S^{1}$. Let us denote by $v$ the infinitesimal generator of the $S^{1}$-action. Then, $\Omega$ is a torsion-free $Spin(7)$-structure if and only if $\tau = \tau_{2}\in\Omega^{2}_{14}\left(\mathcal{M}_{8}/S^{1}\right)$, where $\tau$ denotes the torsion of $\varphi$, and 

\begin{equation}
\label{eq:conditiontorsion}
d v^{\flat} + \pi^{\ast}\tau_{2} = 0\, ,  \qquad d\| v\| = 0 \, .
\end{equation}
\end{prop}

\begin{proof}
Using theorem \ref{thm:bijectionM8}, the induced $G_{2}$-structure on $\mathcal{M}_{8}/S^{1}$ is given by $\pi^{\ast}\left(\varphi\right) = \iota_{w}\Omega = \phi$ (where $\omega=\frac{v}{||v||}$) and we can write $\Omega$ as follows

\begin{equation}
\Omega = w^{\flat}\wedge\phi + \ast\left(w^{\flat}\wedge\phi\right)\, .
\end{equation}

\noindent
The torsion-free condition $d\Omega = 0$ translates into

\begin{equation}
\label{eq:conditionholonomy1}
\left\{ dw^{\flat} - 3w^{\flat}\wedge\tau_{1} + \tau_{2}\right\}\wedge\varphi + \left\{4\tau_{1} - w^{\flat}\tau_{0}  \right\}\wedge \ast_{7}\varphi-w^{\flat}\wedge\ast_{7}\tau_{3} = 0\, ,
\end{equation}

\noindent
where for notational simplicity we omitted the pull-back of all seven-dimensional objects, $\tau_i$ are the torsion classes defined in (\ref{G2torsionclasses}) and we have used 
\begin{equation} \label{oldlemma}
\ast\phi = w^{b}\wedge \pi^{\ast}\left(\ast_{7}\varphi\right) \ ,
\end{equation}
which is straightforward to show. 
Contracting  \eqref{eq:conditionholonomy1} with $w$ and using that the metric $g_{8}$ is $S^{1}$-invariant we obtain

\begin{equation}
-\tau_{3} = 3\ast_{7}\left(\tau_{1}\wedge\varphi\right) +\tau_{0}\varphi\, ,
\end{equation}

\noindent
which implies $\tau_{0} = \tau_{1} = \tau_{3} = 0$ in order for $\tau_{3}\in \Omega^{3}_{27}\left(\mathcal{M}_{8}/S^{1}\right)$. Therefore equation \eqref{eq:conditionholonomy1} translates into

\begin{equation}
dw^{\flat} + \pi^{\ast}\tau_{2} = 0\, .
\end{equation}

\noindent
On the other hand, $\tau_{0} = \tau_{1} = \tau_{3} = 0$ imply that  $\mathcal{L}_{w}\Omega = 0$. However, by assumption $\mathcal{L}_{v}\Omega = 0$ and thus we obtain

\begin{equation}
d\| v\| = 0\, ,
\end{equation}

\noindent
and we conclude.
\end{proof}

\noindent
Note that proposition \ref{prop:M8holonomy} implies that $\varphi$ satisfies

\begin{equation}
d_{7}\varphi = 0\, , \qquad d_{7}\ast_{7}\varphi = \tau_{2}\wedge\varphi\, ,
\end{equation}

\noindent
and thus $\tau_{2}$ is a closed two-form. In fact, since $\tau_{2}\in\Omega^{2}_{14}\left(\mathcal{M}_{8}/S^{1}\right)$ it has to satisfy the constraint

\begin{equation}
\tau_{2}\wedge\varphi + \ast_{7}\tau_{2} = 0\, ,
\end{equation}

\noindent
and hence applying the exterior derivative we obtain that $d\ast\tau_{2} = 0$, so $\tau_{2}$ is closed and coclosed and thus

\begin{equation}
\Delta_{7}\tau_{2} = 0\, ,
\end{equation}

\noindent
that is, $\tau_{2}$ is a harmonic two-form on $\mathcal{M}_{8}/S^{1}$. Furthermore, using that $\tau_{2}$ is closed and coclosed in equation \eqref{eq:conditiontorsion} we arrive at

\begin{equation}
\Delta\, dv^{\flat} = 0\, ,
\end{equation}

\noindent
namely $dv^{\flat}$ is a harmonic two-form on $\mathcal{M}_{8}$. We see then that the requirement of a $Spin(7)$ torsion-free $S^{1}$-invariant structure imposes strong constrains on the induced $G_{2}$-structure, namely the only possible non-zero torsion class is $\tau_2$ and it has to be an harmonic function. 
Indeed, in theorem \ref{thm:bijectionM8} the set-up, although naturally compatible with the $S^{1}$-action on $\mathcal{M}_{8}$, is quite constrained: we start with a $S^{1}$-invariant admissible four-form $\Omega$, which induces a Riemannian metric $g_{8}$, and we rewrite $\Omega$ as

\begin{equation}
\Omega = w^{\flat}\wedge \iota_{w}\Omega + \ast\left(w^{\flat}\wedge \iota_{w}\Omega\right)\, .
\end{equation} 

\noindent
We then require the $G_{2}$ structure $\varphi$ to be given by $\pi^{\ast}\varphi = \iota_{w}\Omega$. This way, we guarantee that $\pi\colon\mathcal{M}_{8}\to \mathcal{M}_{8}/S^{1}$ is a Riemannian submersion with respect to the corresponding induced metric. 

A less restricted scenario would be to relax the $S^1$-invariance, namely still consider  an eight-dimensional $S^{1}$-bundle $\mathcal{M}_{8}\xrightarrow{\pi}\mathcal{M}_{8}/S^{1}$ equipped with an admissible four-form, but the latter is not necessarily $S^{1}$-invariant. Start from an arbitrary $G_{2}$ structure $\tilde{\varphi}$ on $\mathcal{M}_{8}/S^1$,  and we will construct an admissible four-form on $\mathcal{M}_{8}$ as follows

\begin{itemize}

\item Pick a codimension-one distribution $H\subset T\mathcal{M}_{8}$ transverse to $\tilde{v}$, where $\tilde{v}$ is a vector field $C^{\infty}\left(\mathcal{M}_{8}\right)$-proportional to $v$.

\item Equip $\mathcal{M}_{8}$ with the unique Riemannian metric $\tilde{g}_{8}$ such that $\|\tilde{v}\| = 1$, $H = \tilde{v}^{\perp}$ and which makes $\pi\colon\mathcal{M}_{8}\to \mathcal{M}_{8}/S^{1}$ into a Riemannian submersion.

\item Equip $\mathcal{M}_{8}$ with the following admissible four-form

\begin{equation}
\label{eq:tildeOmega}
\tilde{\Omega} = \tilde{v}^{\flat}\wedge \tilde{\phi} + \ast\left(\tilde{v}^{\flat}\wedge \tilde{\phi}\right)
\end{equation}

\noindent
Then, by means of theorem 5.27 of reference \cite{Grana:2014rpa} we have that $\tilde{\Omega}$ induces $\tilde{g}_{8}$ and in general it will not be $S^{1}$-invariant. 

\end{itemize}

\noindent
Notice that are now in exactly the same set-up as in proposition \ref{prop:M8holonomy} but dropping the requirement of $\Omega$ being $S^{1}$-invariant. We obtain the following proposition.

\begin{prop}
\label{prop:M8holonomyII}
Let $\mathcal{M}_{8}\xrightarrow{\pi} \mathcal{M}_{8}/S^{1}$ denote an eight-dimensional oriented spin $S^{1}$-bundle and let us pick a $G_{2}$-structure $\varphi$ on $\mathcal{M}_{8}/S^{1}$. Let us equip $\mathcal{M}_{8}$ with an admissible four-form $\tilde{\Omega}\in\Omega^{4}_{a}\left(\mathcal{M}_{8}\right)$ as in equation \eqref{eq:tildeOmega}, and let $v$ be the infinitesimal generator of the $S^{1}$-action. Then, $\tilde{\Omega}$ is a torsion-free $Spin(7)$-structure if and only if $\tau = \tau_{2}\in\Omega^{2}_{14}\left(\mathcal{M}_{8}/S^{1}\right)$, where $\tau$ denotes the torsion of $\varphi$, and 

\begin{equation}
\label{eq:conditiontorsionII}
d \tilde{v}^{\flat} + \pi^{\ast}\tau_{2} = 0\, ,\qquad \mathcal{L}_{\tilde{v}}\tilde{\Omega} = 0\, . 
\end{equation}
\end{prop}

\begin{proof}
The torsion-free condition $d\tilde{\Omega} = 0$ translates into

\begin{equation}
\label{eq:conditionholonomy}
\left\{ d\tilde{v}^{\flat} - 3\tilde{v}^{\flat}\wedge\tau_{1} + \tau_{2}\right\}\wedge\varphi + \left\{4\tau_{1} - \tilde{v}^{\flat} \tau_{0}\right\}\wedge \ast_{7}\varphi-\tilde{v}^{\flat}\wedge\ast_{7}\tau_{3} = 0\, ,
\end{equation}
where again the pull-back of seven-dimensional objects should be understood.
 Contracting equation \eqref{eq:conditionholonomy} with $\tilde v$ we obtain

\begin{equation}
\label{eq:condition2}
\left\{\iota_{\tilde{v}} d\tilde{v}^{\flat} - 3\tau_{1}\right\}\wedge\varphi + \tau_{0} \wedge\ast_{7}\varphi-\ast_{7}\tau_{3} = 0\, ,
\end{equation}

\noindent
which implies

\begin{equation}
\label{eq:solutionhol}
\iota_{\tilde{v}} d\tilde{v}^{\flat} - 3\pi^{\ast}\tau_{1} = 0\, ,\qquad \tau_{0} = \tau_{3} = 0\, .
\end{equation}

\noindent
 However, plugging this back into  \eqref{eq:conditionholonomy} we obtain that $\tau_{1} = 0$  and thus

\begin{equation}
d\tilde{v}^{\flat} + \pi^{\ast}\tau_{2} = 0\, , \qquad \iota_{\tilde{v}} d\tilde{v}^{\flat} = 0\, ,
\end{equation}

\noindent
which implies that $\mathcal{L}_{\tilde{v}}\tilde{v}^{\flat} = 0$ and thus $\mathcal{L}_{v}\tilde{g}_{8} = 0$. Furthermore, since $\tau_{0} = \tau_{1} = \tau_{3} = 0$ we have that $d\phi = 0$ and therefore $\mathcal{L}_{\tilde{v}}\tilde \Omega = 0$, although in general $\tilde{\Omega}$ it is not $S^{1}$-invariant. Therefore, we can write the condition of $Spin(7)$-holonomy as

\begin{equation}
d\tilde{v}^{\flat} + \pi^{\ast}\tau_{2} = 0\, , \qquad \mathcal{L}_{\tilde{v}}\tilde{\Omega} = 0\, .
\end{equation}
\end{proof}

\noindent
Therefore we see then that $Spin(7)$-holonomy implies $\mathcal{L}_{\tilde{v}}\tilde{\Omega} = 0$ and when $\tilde{v} = v$ we go back to proposition \ref{prop:M8holonomy} and we finally obtain that $\Omega$ is $S^{1}$-invariant. 

From the supersymmetry conditions relating fluxes to torsion classes obtained in \cite{Kaste:2003zd}, it seems that a supersymmetric compactification to Minskowski with torsion given only by $\tau_2$ is possible. We leave more explicit constructions with more general $G_{2}$-torsion classes, and thus involving higher-dimensional manifolds, for future work.


\section{Relating moduli spaces}
\label{sec:modulispaces}


In this section we give an idea of how to extract  information about the moduli space of supersymmetric $G_{2}$-structures on a seven-dimensional manifold from the moduli space of torsion-free $Spin(7)$-structures on an eight-dimensional one. It is intended to be only descriptive. 

Recall that (see definition \ref{def:susyG2}), given a seven-dimensional oriented and spin manifold $\mathcal{M}_{7}$, a supersymmetric $G_2$-structure is defined to be a quadruplet $\left(\varphi,G,e^{2A},\mu\right)$ satisfying a set of differential conditions coming from the Killing spinor equation \eqref{eq:N1susy11}. We will denote by $\mathcal{B}$ the set of all supersymmetric $G_{2}$-structures $\left(\varphi,G,e^{2A},\mu \right)$ on a given manifold $\mathcal{M}_7$. 

\begin{definition}
The moduli space $\mathcal{X}_{7}$ of supersymmetric $G_{2}$-structures on $\mathcal{M}_{7}$ is the space $\mathcal{B}$ of all supersymmetric $G_{2}$-structures modulo the appropriate equivalence relation, namely

\begin{equation} 
\label{chi7}
\mathcal{X}_{7} = \left\{\left(\varphi,\mathsf{G},e^{2A},\mu\right) \in\mathfrak{B}\right\}/\tilde{\mathcal{D}}_{7}\, ,
\end{equation}

\noindent
where $\tilde{\mathcal{D}}_{7}$ relates equivalent supersymmetric $G_{2}$-structures, to wit

\begin{equation}
\left(\varphi,G,e^{2A} ,\mu\right)\sim \left(\tilde{\varphi},\tilde{G}, e^{2\tilde A} ,\tilde \mu \right)
\end{equation}

\noindent
if they are related by a diffeomorphism or by a gauge transformation of $G$. We will denote the elements of $\mathcal{X}_{7}$, namely the corresponding equivalence classes, by $\left[\varphi,G,e^{2A},\mu \right]$.
\end{definition}

\noindent
The moduli space $\mathcal{X}_{7}$ of supersymmetric $G_{2}$-structures remains as a largely unexplored space, and, to the best of our knowledge it has not even been  proven if it is a finite-dimensional manifold. We are going to argue that such space can be related, at least in some special instances, to the moduli space of torsion-free structures on a higher-dimensional manifold, which is a space much more under control. Indeed, this procedure can be justified to work from the exceptional generalized geometry formulation of the moduli space; we will come back again to this point at the end of the section. For now, we will illustrate the procedure in the simplest case where the higher-dimensional manifold is eight-dimensional and thus it can be identified with an intermediate manifold, keeping in mind that for higher-dimensional manifolds the construction proceeds along similar lines. In fact, the construction presented in section \ref{sec:intermediateEGG} already suggests to consider instead of the space of supersymmetric $G_{2}$-structures an eight-dimensional manifold $\mathcal{M}_{8}$ constructed from $\mathcal{M}_{7}$ in such a way that a supersymmetric $G_{2}$-structure on $\mathcal{M}_{7}$ induces a torsion-free $Spin(7)$-structure on $\mathcal{M}_{8}$. Let us just very briefly recall some basic facts about the moduli space of torsion-free $Spin(7)$-structures on $\mathcal{M}_{8}$.

Given an eight-dimensional, compact, oriented manifold ${\cal M}_8$, the moduli space of torsion-free $Spin(7)$-structures on $\mathcal{M}_{8}$ is given by \cite{Joyce2007}

\begin{equation}
\mathcal{X}_{8} \equiv \mathcal{N}_{8}/\mathcal{D} = \left\{\Omega\in\Omega^{4}_{a}\left(\mathcal{M}_{8}\right)\,\, | \,\, d\Omega = 0 \right\}/\mathcal{D}\, ,
\end{equation} 

\noindent
where $\Omega^{4}_{a}\left(\mathcal{M}\right)$ is the set of admissible forms on $\mathcal{M}_{8}$ (see footnote \ref{foot:admissible}), or in other words the space of sections of the bundle $\mathcal{A}^4\mathcal{M}\to\mathcal{M}_{8}$ with fibre $Gl_{+}\left(8,\mathbb{R}\right)/Spin(7)$, and $\mathcal{D}$ denotes de group of diffeomorphisms of $\mathcal{M}_{8}$ isotopic to the identity. $\mathcal{X}_{8}$ is a smooth manifold of dimension $\dim\, \mathcal{X}_{8} = b_{1}^4 + b_{7}^4 + b_{35}^4$, where $b_{1}$, $b_{7}$ and $b_{35}$ denote the betti numbers of the cohomology group of 4-forms corresponding to the $Spin(7)$-representation indicated by the subscript. 

We want to use ${\mathcal X}_8$ to infer properties about ${\mathcal X}_7$. For that, we have shown that one can 
construct the eight-dimensional manifold $\mathcal{M}_{8}$ from a seven-dimensional one $\mathcal{M}_{7}$ such that from every  supersymmetric $G_{2}$-structure on $\mathcal{M}_{7}$ one can construct a topological $Spin(7)$-structure $\Omega$ on $\mathcal{M}_{8}$. We have thus a map

\begin{equation}
\label{eq:spacemap}
\mathfrak{I}_{Top}\colon  {\mathfrak{B}} \to \Omega_{a}\left(\mathcal{M}_8\right) \, .
\end{equation}

\noindent
from the set of supersymmetric $G_{2}$-structures on $\mathcal{M}_{7}$ to the set of topological $Spin(7)$-structures on $\mathcal{M}_{8}$. Now, in order to relate ${\mathfrak{B}}$ to a known moduli space we demand the image of $\mathfrak{J}_{Top}$ to be the subset  $\mathcal{N}_{8}\subset\Omega_{a}\left(\mathcal{M}_8\right)$ given by the torsion-free $Spin(7)$-structures, whose moduli space is well under control. We define thus $\mathfrak{I}_{Hol}$, the corresponding restriction of $\mathfrak{I}_{Top}$, as follows

\begin{equation}
\label{eq:spacemapII}
\mathfrak{I}_{Hol} \equiv\mathfrak{I}_{Top}|_{\mathfrak{B}_{Hol}}\colon\mathfrak{B}_{Hol}\to \mathcal{N}_8 \, ,
\end{equation}

\noindent
where $\mathfrak{B}_{Hol} = \mathfrak{I}^{-1}_{Top}\left(\mathcal{N}_{8}\right)\subset\mathfrak{B}$ is the restricted set of supersymmetric $G_{2}$-structures that induce through $\mathfrak{I}_{Top}$ a torsion-free $Spin(7)$-structure on $\mathcal{M}_{8}$. We presented explicit examples of $\mathfrak{I}_{Hol}$ of this in sections \ref{sec:Spin7cone} and \ref{sec:Spin7bundle}. For example, in the $Spin(7)$-cone construction, for every nearly $G_{2}$ structure on $\mathcal{M}_{7}$ one obtains a torsion-free $Spin(7)$ structure on $\mathcal{M}_{8}$, although in this case the $\mathcal{M}_{8}$ turns out to be non-compact. 

The last step consists in evaluating if the map \eqref{eq:spacemapII} descends to a well defined map on the corresponding equivalence classes. Assuming that it is the case, then we obtain a new map

\begin{equation}
\mathfrak{I}\colon\mathfrak{B}_{Hol}/\tilde{D}\to \mathcal{X}_{8}\, ,
\end{equation}

\noindent
between the moduli space $\mathfrak{B}_{Hol}/\tilde{D}$ of restricted supersymemtric $G_{2}$-structures on $\mathcal{M}_{7}$ and the moduli space of torsion-free $Spin(7)$-structures on $\mathcal{M}_{8}$. From the particular properties of this map, one may be able to extract information about $\mathfrak{B}_{Hol}/\tilde{D}$ taking into account what is known about $\mathcal{X}_{8}$. The first property to be checked is the continuity and injectivity of the map. Remarkably enough, if one can proove that $\mathfrak{I}$ is a local homeomorphism between topological spaces, then $\mathfrak{B}_{Hol}/\tilde{D}$ would be a topological manifold of the same dimension as $\mathcal{N}_{8}$, namely $b_{1}^4 + b_{7}^4 + b_{35}^4$. That is, we would have obtained the dimension of the moduli space of restricted supersymmetric $G_{2}$-structures on $\mathcal{M}_{7}$ in terms of the topological data of $\mathcal{M}_{8}$.

For supersymmetry reasons, the moduli space $\mathcal{X}_{8}$ must be identified with the K\"ahler-Hodge manifold appearing in the non-linear sigma model of four-dimensional $\mathcal{N}=1$ Supergravity. A K\"ahler-Hodge manifold is a K\"ahler manifold such that the K\"ahler form represents an integer cohomology class, and therefore is the curvature form of the associated $U(1)$-bundle. There are several obstructions on a manifold to be K\"ahler-Hodge, the most obvious one being that its dimension must be even. Therefore, if the map $\mathfrak{I}$ exists then the dimension of $\mathcal{X}_{8}$ must be even, and we conclude that there are as well topological obstructions to the existence of $\mathfrak{I}$.

The example of the $Spin(7)$-conifold is again useful. In that case, there is a one-to-one correspondence between nearly $G_{2}$-structures on $\mathcal{M}_{7}$ and torsion-free $Spin(7)$-structures on $\mathcal{M}_{8}$. However, the resulting $\mathcal{M}_{8}$ is non-compact so there are not as many available results about its moduli space as there are in the compact case. In fact, we have been talking about an eight-dimensional manifold $\mathcal{M}_{8}$ constructed from $\mathcal{M}_{7}$. Ideally we would like this $\mathcal{M}_{8}$ to be compact, since in that case the moduli space of torsion-free $Spin(7)$-structures is under control. However, as we already explained, this is a too optimistic scenario, since constructing a compact $Spin(7)$-manifold $\mathcal{M}_{8}$ from $\mathcal{M}_{7}$ is an extremely difficult task, for instance there are no isometries that may allow for an adapted system of coordinates to write the metric. Still, constructing a non-compact $\mathcal{M}_{8}$ is interesting, because depending on its asymptotics, some results are also known about the moduli space, in particular if it is a smooth manifold of a finite dimensionality, which would be the first thing to check for the moduli space of supersymmetric $G_{2}$-structures on $\mathcal{M}_{7}$. 

The procedure explained here refers to seven-dimensional manifolds equipped with a supersymmetric $G_{2}$-structure which induces a torsion-free $Spin(7)$-structure in an eight-dimensional manifold $\mathcal{M}_{8}$. It is to be expected that if the supersymmetric $G_{2}$-structure is general enough, i.e. if it has many nonzero torsion classes, then there is no way it induces a torsion-free $Spin(7)$ structure on any eight-dimensional manifold, namely the inclusion $\mathfrak{B}_{Hol}\subset\mathfrak{B}$ is strict. It is natural to ask under which conditions one should expect $\mathfrak{B}_{Hol} =\mathfrak{B}$, namely, when we can \emph{geometrize} the full space of supersymemtric $G_{2}$-structures. We expect that considering a $n>\, 8\,\,$- dimensional manifold of special holonomy and doing the same procedure, this time imposing that the supersymmetric $G_{2}$-structure in seven dimensions induces a torsion-free structure in $n$-dimensions, which can be for example a $SU(\frac{n}{2})$ or $Sp(\frac{n}{4})$ torsion-free structure, one might be able to completely geometrize the most general supersymmetric $G_{2}$-structure. This is an open question that can be formulated as follows

\begin{itemize}

\item Is there a $n$-dimensional principal fibre bundle $\mathcal{M}$ with fibre $G$ and seven-dimensional base $\mathcal{M}/G$ such that every $G_{2}$-structure on $\mathcal{M}_{7}$ induces a torsion-free structure on $\mathcal{M}$? 

\end{itemize}

\noindent
In the eight-dimensional case  $\mathcal{M}_{8}$ corresponds to an intermediate manifold. The question is now, if we choose to embed the supersymmetric $G_{2}$-torsion in a higher dimensional manifold, is there again a connection to exceptional generalized geometry as there is in the eight-dimensional case? The answer is afirmative. Recall that the definition \ref{def:intermediatemanifolds} of an intermediate manifold was based on the decomposition of the ${\bf 912}$ of $E_{7(7)}\times\mathbb{R}^{+}$ in terms of $Sl(8,\mathbb{R})$ representations, which, acting on an eight-dimensional vector space required an eight-dimensional manifold. However, we can consider other subgroups $G\subset E_{7(7)}\times\mathbb{R}^{+}$ and follow the same steps defining now a $n$-dimensional manifold, where $n$ is equal to the dimension of the representation of $G$, equipped with the tensors $\mathfrak{S}_{n}$ appearing in the corresponding branching rule. This way, we can translate again all the information of the moduli space to a $n$-dimensional manifold $\mathcal{M}_{n}$ equipped with the appropriate tensors, and impose a torsion-free condition on them. Imposing the torsion-free condition will again restrict the moduli space to the appropriate subspace, but since the manifold is of a higher dimension, the restriction will be less severe than on the eight-dimensional case. Therefore, in some sense, we have a dictionary to go from the description of the moduli space given in exceptional generalized geometry, to the particular geometry of a $n$-dimensional manifold $\mathcal{M}_{n}$ whose dimension and geometric characteristics are determined by a subgroup $G$ of $G\subset E_{7(7)}\times\mathbb{R}^{+}$. In this note however we have focused for simplicity only on the eight-dimensional case.


\section{Summary}

 
It is convenient to give a brief summary of the main results of the paper. 

\begin{itemize}

\item Using Exceptional Generalized Geometry, the moduli space of M-theory $\mathcal{N}=1$ compactifications to four dimensions can be described through the space of sections of a particular $E_{7(7)}\times\mathbb{R}^{+}$-bundle over the internal space $\mathcal{M}_{7}$. We propose that, using a subgroup $G\subset E_{7(7)}\times\mathbb{R}^{+}$, all the information of this bundle can be translated to an $n$-dimensional manifold $\mathcal{M}_{n}$, using appropriate sections $\mathfrak{S}$ of its tensor bundles, determined by the branching rule $E_{7(7)}\times\mathbb{R}^{+}\to G$. Here $n$ is the dimension of the representation of $G$.

\item Since the moduli space of $\mathfrak{S}$-tensors on $\mathcal{M}_{n}$ corresponds to the moduli space of M-theory compactifications, we propose the existence of a map $\mathfrak{I}$ that relates both spaces and which can be used to study the geometry of the latter from the geometry of the former.

\item In order to fully specify the geometry of the moduli space of $\mathfrak{S}$-tensors, we need to know the supersymmetry conditions in the exceptional generalized language, something that is not currently available in the literature. Therefore, we consider the case where $\mathfrak{S}$ induces a classical torsion-free structure on $\mathcal{M}_{n}$, when $n$ is eight. The moduli space of classical torsion-free structures is well under control, so it is reasonable that given the map $\mathfrak{I}$ one can extract from it information about the moduli space of $\mathcal{N}=1$ M-theory compactifications.

\item We give explicit examples of the previous abstract constructions for the case where $n$ is eight and $\mathcal{M}_{8}$ is either a product manifold or an $S^{1}$-bundle over the internal space $\mathcal{M}_{7}$, fully characterizing the torsion classes on the base when the total space is equipped with a torsion-free $Spin(7)$-structure.

\item We have proven that every eight-dimensional, oriented, spin $S^{1}$-bundle admits a topological $Spin(7)$-structure.
\end{itemize}

\noindent
We have found that by imposing $Spin(7)$ holonomy on an eight-dimensional manifold built either as a product manifold or as a spin $S^1$-bundle, one can describe situations in which the $G_2$ torsion classes are respectively $\tau_0$ and $\tau_2$.  We summarize our results in the following table

\begin{table}[ht]
   \begin{center}
      \bigskip
      \begin{tabular}{|c|c|c|c|} \hline \rule[-0.3cm]{0cm}{0.8cm}
       ${\mathcal M}_8$  & $\Omega_4$ & $g_8$ & $\tau({\mathcal M}_7)$ \\ \hline
       ${\mathbb R}^+ \times {\mathcal M}_7$ &  $ dt \wedge\varphi + (\tau_0 t) \,\ast_{7}\varphi$ & $dt \otimes dt  + (\tau_0 t)^2 g_{7}$ & $\tau_0$, $d_7\tau_0=0$ 
    \\ \hline 
         \rule[-0.3cm]{0cm}{0.8cm} 
       $S^{1}\hookrightarrow {\mathcal M}_8\, , \mathcal{M}_{7}\simeq\mathcal{M}_{8}/S^{1}$  & $v^{\flat}\wedge\phi + \ast\left(v^{\flat}\wedge\phi\right)$ & $v^{\flat} \otimes v^{\flat} + \pi^{\ast}g_7$ &   $\tau_2=-d_7 v^{\flat}$ \\ \cline{2-4}
         \rule[-0.3cm]{0cm}{0.8cm}  
         & $\tilde v^{\flat}\wedge\phi + \ast\left(\tilde v^{\flat}\wedge\phi\right)$ & ${\tilde v}^{\flat}\otimes {\tilde v}^{\flat} +\pi^{\ast} g_7$ &   $\tau_2=-d_7 \tilde v^{\flat} $\\ \hline
            \end{tabular}
      \label{summary}
      \caption{Summary of results. In the $S^1$ bundle cases, $v$ is the generator of the $S^1$ action, and $\tilde v=f v$, for some function $f\in C^{\infty} ({\mathcal M}_8)$.} 
   \end{center}
\end{table}

\noindent
To embedd other torsion classes (or in other words to geometrize fluxes in other representations) via a torsion-free structure, one has to either build the eight-dimensional manifold in a different way (like for example as a codimension-one foliation), or consider torsion-free structures in higer-dimensions than eight, which would correspond to decompositions of $E_{7(7)}\times\mathbb{R}^{+}$ respect to subgroups $G$ with representations of dimension higher than eight.


\acknowledgments

We would like to thank Dominic Joyce, Spiro Karigiannis, Ruben Minasian, Simon Salamon, Raffaele Savelli and Marco Zamon for very useful discussions and comments. We are particularly grateful to Dominic Joyce and Marco Zamon for reading preliminary versions of the manuscript and making very important suggestions. This work was supported in part by the ERC Starting Grant 259133 -- ObservableString, and by projects MTM2011-22612.


\appendix


\renewcommand{\leftmark}{\MakeUppercase{Bibliography}}
\phantomsection
\bibliographystyle{JHEP}
\bibliography{References}
\label{biblio}
\end{document}